\documentclass{llncs}
\usepackage{epsfig}
\usepackage{amsmath}
\usepackage{amssymb}

\newcommand{\set}[2]{\{#1\,|\,#2\}}
\newcommand{\tuple}[1]{\left<{#1}\right>}
\newcommand{\Dom}{\mbox{\sl dom\/}}

\DeclareMathOperator{\maxpref}{\sqcap}

\def\bbbr{{\rm I\!R}}

\begin{document}

\title{Myhill-Nerode Relation for Sequentiable Structures}
\author{Stefan Gerdjikov\inst{1,2} \and Stoyan Mihov\inst{2}}

\institute{
Faculty of Mathematics and Informatics\\
Sofia University\\
5, James Borchier Blvd., Sofia 1164, Bulgaria\\
stefangerdzhikov@fmi.uni-sofia.bg
\and
Institute of Information and Communication Technologies\\
Bulgarian Academy of Sciences\\
25A, Acad. G. Bonchev Str., Sofia 1113, Bulgaria\\
stoyan@lml.bas.bg
}

\maketitle

\begin{abstract}
Sequentiable structures are a subclass of monoids that generalise 
the free monoids and the monoid of non-negative real numbers with addition.
In this paper we consider functions $f:\Sigma^*\rightarrow {\cal M}$ and
define the Myhill-Nerode relation for these functions. We prove that
a function of finite index, $n$, can be represented with a subsequential
transducer with $n$ states.
\end{abstract}

\section{Introduction}
Finite-state automata and transducer are widely used in many areas and applications of computer science \cite{HMU01}. Subsequential finite-state transducers are highly computationally efficient for language modelling and text processing tasks, and thus provide a very desirable technique in computational linguistics \cite{Mo95,Mo97,RS97}.

Sequentail structures were introduced in~\cite{LATA2017}. They generalise the notion of free monoids and non-negative real numbers with addition. In~\cite{LATA2017} we characterised the subsequential rational functions that map words to sequentiable monoids. In this paper we study the theoretical foundations for the minimisation of subsequential rational functions over sequentiable structures. We generalise the notion of Myhill-Nerode relation for rational structures over free monoids and real numbers,~\cite{Mohri00}, to the case of sequentiable structures. Our main contribution is Theorem~\ref{th:Nerode-Myhill} which proves that rational functions of finite index are sequentiable.
\section{Formal Preliminaries}
We assume that the reader is familiar with the basic notions of an alphabet and monoid, see~\cite{Eil74}. 

\begin{definition}\label{semi-lattice}
${\cal L}=\tuple{L,\le,\maxpref}$ is called a meet semi-lattice if $\tuple{L,\le}$ is a partially ordered set
and for any two elements $a,b\in L$, $c=a\maxpref b\in L$ is the biggest element w.r.t. $\le$ such that 
$c\le a$ and $c\le b$.
\end{definition}

\begin{definition}\label{DefSubsequentialTransducer}
A {\em monoidal subsequential finite-state transducer} is a tuple 
${\cal T}=\tuple{\Sigma,{\cal M},Q,q_0,F,\delta,\lambda,\iota,\Psi}$ 
where:
\begin{itemize}
\item
$\Sigma$ is a finite alphabet,
\item
  ${\cal M}=\tuple{M,\circ,e}$ is a monoid, 
\item
  $Q$ is a finite set of states, 
\item
  $q_0\in Q$ is the initial state, 
\item 
  $F\subseteq Q$ is the set of final states,  
\item
  $\delta :  Q\times \Sigma \rightarrow Q$ 
  is (possibly partial) function called transition function,
\item
  $\lambda :  Q\times \Sigma \rightarrow M$ 
  is a function with domain, $\Dom(\lambda)=\Dom(\delta)$, called the transition output function,
\item $\iota\in M$, initial output, and
\item
$\Psi : F \rightarrow {\cal M}$ is the state output function.
\end{itemize}

The {\em generalised transition function} $\delta^\ast$ is defined as the
inclusion-wise least function on $Q \times \Sigma^\ast \rightarrow Q$ with the following closure 
properties:
\begin{itemize}
\item
  for all $q\in Q$ we have $\delta^\ast(q,\varepsilon)=q$.
\item
  For all $q\in Q, \alpha \in \Sigma^\ast$ and $a \in \Sigma$ : $\delta^\ast(q,\alpha a) = \delta(\delta^\ast(q,\alpha),a)$.
\end{itemize}
The {\em generalised transition output function function} $\lambda^\ast$ is the 
inclusion-wise least function on $Q \times \Sigma^\ast \rightarrow M$ with the following closure 
properties:
\begin{itemize}
\item
  for all $q\in Q$ we have $\lambda^\ast(q,\varepsilon)=e$.
\item
  For all $q\in Q, \alpha \in \Sigma^\ast$ and $a \in \Sigma$ : $\lambda^\ast(q,\alpha a) = \lambda^\ast(q,\alpha) \circ \lambda(\delta^\ast(q,\alpha),a)$.
\end{itemize}

The {\em output function} represented by the subsequential finite-state transducer ${\cal T}$ is  ${\cal O}_{\cal T} : \Sigma^\ast \rightarrow M$ defined as: ${\cal O}_{\cal T}(\alpha)=\iota\lambda^\ast(q_0,\alpha) \circ \Psi(\delta^\ast(q_0,\alpha))$ if $\delta^\ast(q_0,\alpha) \in F$.
\end{definition}

\section{Sequentiable Structures}
Pre-sequentiable and sequentiable structures were introduced in~\cite{LATA2017}. Here we give an equivalent definition based
on the following relation of monoidal elements.
\begin{definition}
For a monoid ${\cal M}$ we define the relation $\le_M$ on ${\cal M}$ as:
\begin{equation*}
a\le_M b \iff \exists c\in M (a c =b).
\end{equation*}
\end{definition}
\begin{definition}
A tuple ${\cal M} = \tuple{M,\circ,e,\maxpref,\|.\|}$ is a {\em pre-sequentiable structure} if
\begin{enumerate}
\item $\tuple{M,\circ,e}$ is a monoid, which supports left cancelation,
\item $\tuple{M,\le_M,\maxpref}$ is a meet semi-lattice.
\item The function $\|.\| : M \rightarrow \bbbr^+$ called {\em norm} is a homomorphism of the monoids $\tuple{M,\circ,e}$ and $\tuple{\bbbr^+,+,0}$, which maps only $e$ to $0$. I.e.  $\|a\| = 0 \rightarrow a=e$.
\end{enumerate}
\end{definition}

\begin{definition}\label{DefinitionSequentialStructure}
A pre-sequentiable structure  ${\cal M}= \tuple{M,\circ,e,\maxpref,\|.\|}$ is called {\em sequentiable structure} if the condition: 
\begin{equation}\label{Monotone}
( a\le_M c \, \& \, b\le_M d\, \& \, a\maxpref b =e \rightarrow c\maxpref d=e)
\end{equation}
holds for all $a,b,c,d\in M\setminus\{e\}$.
\end{definition}
In case of sequentiable structures we additionally get the right cancellation property and Levy-like lemma:  
\begin{proposition}\label{RightCancellation}
Let ${\cal M}=\tuple{M,\circ,e,\maxpref,\|.\|}$ be a sequentiable structure.
\begin{enumerate}
\item If $\|a\| \le \|b\|$ and $a\le_M bc$, then $a\le_M b$.
\item $\forall a,b,c\in M (a\circ c=b\circ c\rightarrow a=b)$ (right cancellation property),
\item for all $a_1,a_2,b_1,b_2 \in M$ with $a_1 \circ a_2 = b_1 \circ b_2 \ \&\ \|b_1\| \ge \|a_1\|$ there is a $c$
such that $a_1 c=b_1$ and $cb_2=a_2$.
\end{enumerate}
\end{proposition}
\begin{proof}
1. Let $m=a\maxpref b$. We prove that $m=a$. Indeed, $a=ma_1$ and $b=mb_1$.
By the left cancellation property, we have that $a_1\le_M b_1 c$.
Next, since $\|a\|\le \|b\|$ we get that $\|a_1\|\le \|b_1\|$. Assume that $a_1\neq e$. Thus,
$\|a_1\|>0$ and therefore, $\|b_1\|\neq 0$, in particular $e<_M b_1$. This shows, that $e<_M a_1\le b_1 c$
and $e<_M b_1\le b_1c$. Since $a_1\maxpref b_1=e$, by Equation~\ref{Monotone},
we conclude that $b_1c=e$ which is a contradiction. Therefore, $a_1=e$ and thus $a=m\le_M b$.

2. By the properties of the homomorphism, we have $\|a\|=\|b\|$. Now the equation $ac=bc$ implies
that $a\le_M bc$ and by Point 1, $a\le_M b$. Similarly, $b\le_M a$. Since $\le_M$ is a partial order,
$a=b$.

3. We have that $a_1\le_M b_1 b_2$. Since $\|a_1\|\le \|b_1\|$, by Point 1, we conclude $a_1\le_M b_1$.
Hence, there is some $c$ with $a_1 c=b_1$. Therefore by the left cancellation $a_2=cb_2$.
\end{proof}

\section{Myhill-Nerode Relation}
In this section we extend in a natural way the definition of the Nerode-Myhill relation for rational functions over free monoids,~\cite{Mohri00}, to arbitrary sequentiable functions. 

\subsubsection*{The Myhill-Nerode relation for subsequential finite-state transducers}
\begin{definition}\label{DefMyhillNerodeRelationSubseq}
Let $f : \Sigma^\ast \rightarrow M$ be a (partial) function. Then 
\[
\begin{array}{ll}
R_f = \set{\tuple{u,v}\in \Sigma^\ast\times \Sigma^\ast} {& \exists u', v' \in M\ \forall w\in\Sigma^\ast:  \\
&(u\cdot w \in \Dom(f) \leftrightarrow v\cdot w \in \Dom(f))\ \& \\
 & (u\cdot w \in \Dom(f) \rightarrow u'^{-1} f(u\cdot w) = v'^{-1} f(v \cdot w)) }
\end{array}
\]
is called the {\em Myhill-Nerode relation}\index{Myhill-Nerode relation} for $f$. 
\end{definition}
Note that $R_f$ acts on the full set $\Sigma^\ast\times \Sigma^\ast$ despite of the fact that $f$ can be partial. 
\begin{proposition}
Let $f : \Sigma^\ast \rightarrow M$ be a function. Then the Myhill-Nerode relation for $f$
is a right invariant equivalence relation. 
\end{proposition}
\begin{proof}
Clearly $R_f$ is an equivalence relation.
 Let $u\ R_f\ v$ and $z \in \Sigma^\ast$. In order to prove that  $u \cdot z \ R_f\ v \cdot z$ we have to show that there exist $u', v' \in M$ such that for any $w \in \Sigma^\ast$ we have (a) $(u \cdot z) \cdot w \in \Dom(f)$ iff $(v \cdot z) \cdot w \in \Dom(f)$, and (b) if $(u \cdot z) \cdot w \in \Dom(f)$, then $u'^{-1} f((u \cdot z) \cdot w) = v'^{-1} f((v \cdot z) \cdot w)$. Let $w \in \Sigma^\ast$. Since $u\ R_f\ v$ there exist $u', v' \in M$ such that for $w' = z \cdot w$ we have (a) $u \cdot w' = (u \cdot z) \cdot w \in \Dom(f)$ iff $v \cdot w' = (v \cdot z) \cdot w \in \Dom(f)$, and (b) if  $u \cdot w' = (u \cdot z) \cdot w \in \Dom(f)$, then  $u'^{-1} f(u \cdot w') = u'^{-1} f((u \cdot z) \cdot w) =  v'^{-1} f(v \cdot w') = v'^{-1} f((v \cdot z) \cdot w)$.\qed
\end{proof}
%

\section{Generalised Myhill-Nerode Characterisation}
In this section we prove that the converse is also true. That is we have the following theorem:
\begin{theorem}\label{th:Nerode-Myhill}
Let $\equiv_{R_f}$ have index $n\in \mathbb{N}$. Then
there is a subsequential transducer with initial output, ${\cal T}$, 
with $n$ states such that $O_{\cal T}=f$.
\end{theorem}

First we note some simple properties of the relation $\equiv_{R_f}$.
\begin{definition}
Given a function $f:\Sigma^*\rightarrow M$ we say that $\alpha\in \Sigma^*$ is $f$-essential 
if there is some $w$ s.t. $f(\alpha w)$ is defined.
\end{definition}

\begin{definition}
We say that a triple $\tuple{m_{\alpha},m_{\beta},s}$ 
where $m_{\alpha}, m_{\beta}\in M$, and $s:\Sigma^*\rightarrow M$ is a partial function
witnesses for $\alpha\equiv_{R_f} \beta$ if and only if:
\begin{enumerate}
\item $\Dom(s) = \{ w\,|\, \alpha w\in \Dom(f)\}$, and
\item for each $w\in \Sigma^*$ s.t. $\alpha w\in \Dom(f)$:
\begin{equation*}
f(\alpha \gamma)=u_{\alpha}s(w) \text{ and } f(\beta\gamma)=u_{\beta}s(w)
\end{equation*}
\end{enumerate}
\end{definition}

\begin{lemma}\label{lemma:Nerode-Myhill}
Let $C_i=[\alpha_i]_{\equiv_{R_f}}$ for
$i=1,\dots, n$ be all the equivalence classes of $\equiv_{R_f}$. Then
there are partial functions $\widehat{s}_i:\Sigma^* \rightarrow M$ with the following two properties:
\begin{enumerate} 
\item for any $\beta\in C_i$ there is a witness $\tuple{m_1,m_2,\widehat{s}_i}$ for $\alpha_i\equiv_{R_f} \beta$.
\item for any $i,j\le n$ and character $a\in \Sigma$ with $\alpha_i a \equiv_{R_f} \alpha_j$ there is a monoid element
$m_{i,a,j}\in M$ with:
\begin{equation*}
\widehat{s}_i(aw) =m_{i,a,j}\widehat{s}_j (w) \text{ whenever } \alpha_j w \in \Dom(f).
\end{equation*}
\end{enumerate}
\end{lemma}

First we show how the Nerode-Myhill-like Theorem~\ref{th:Nerode-Myhill} follows from Lemma~\ref{lemma:Nerode-Myhill}:
\begin{proof}[ of Theorem~\ref{th:Nerode-Myhill}]
Let $C_i=[\alpha_i]_{\equiv_{R_f}}$ for $i=1,\dots,n$ be all the equivalence classes of $\equiv_{R_f}$.
Let $\widehat{s}_i$ and $m_{i,a,j}$ satisfy the conclusion of Lemma~\ref{lemma:Nerode-Myhill}. W.l.o.g. we assume that $C_1=[\varepsilon]_{\equiv_{R_f}}$. If $\Dom(f)=\emptyset$ then clearly $n=1$ and one can construct a trivial subsequential transducer representing $f$. 

Alternatively, let $w\in \Dom(f)$. By the properties of $\widehat{s}_1$ there is a witness $\tuple{m_1,m_{\varepsilon},{\widehat{s}_1}}$ for $\alpha_1\equiv_{R_f} \varepsilon$. We define ${\cal T}$ as:
\begin{equation*}
{\cal T} = \tuple{\Sigma\times M,\{C_i\}_{i=1}^n,C_1,F,\delta,\lambda,m_{\varepsilon},\Psi}, \text{ where }
\end{equation*}
\begin{eqnarray*}
  \quad F & =  &\{C_i \,|\, \varepsilon \in \Dom(s_i)\} \\
  \psi(C_i) &=&  \widehat{s}_i(\varepsilon)\\
\delta(C_i,a) & = & C_j \iff \alpha_i a\equiv_{R_f} \alpha_j \\
\lambda(C_i, a) & = & m_{i,a,j} \iff \delta(C_i,a)=C_j 
\end{eqnarray*}
Since $\equiv_{R_f}$ is right-invariant it is obvious that $\alpha\in C_i$ iff $\delta^*(C_1,\alpha)=C_i$. Since, $C_i \subseteq \Dom(f)$ if and only if $s_i(\varepsilon)$ is defined, we conclude that $\Dom(O_{\cal T})=\Dom(f)$. It remains to show that for any $w\in \Dom(f)$, $O_{\cal T}(w)=f(w)$. Let $w=uv$ be an arbitrary decomposition of $w$ into two words, $u,v\in \Sigma^*$. Let $C_i=\delta^*(C_1,u)$. We prove that:
\begin{equation*}
f(w) = m_{\varepsilon} \lambda^*(C_1,u) \widehat{s}_i(v)
\end{equation*}
by induction on $|u|$. For $|u|=0$, we have that $\tuple{m_1,m_{\varepsilon},{\widehat{s}_1}}$ is a witness for $\alpha_1\equiv_{R_f} \varepsilon$. Since $w\in \Dom(f)$ we conclude that $f(w) =m_{\varepsilon} \widehat{s}_1(w)=m_{\varepsilon} \widehat{s}_1(w)$.
For the induction step, let $w=uv=u a v'$ where $a\in \Sigma$. Let $C_i=\delta^*(C_1,u)$ and $C_j=\delta(C_i,a)$.
By the induction hypothesis we have:
\begin{equation*}
f(w) =m_{\varepsilon} \lambda^*(C_1,u) \widehat{s}_i(v)=m_{\varepsilon} \lambda^*(C_1,u) \widehat{s}_i(av')=m_{\varepsilon} \lambda^*(C_1,u)m_{i,a,j} \widehat{s}_j(v'),
\end{equation*}  
where the last equality follows by the properties of $\widehat{s}_i$, $\widehat{s}_j$, and $m_{i,a,j}$. The last equality is equivalent to:
\begin{equation*}
f(w) = m_{\varepsilon} \lambda^*(C_1,u)m_{i,a,j} \widehat{s}_j(v') =m_{\varepsilon} \lambda^*(C_1,ua) \widehat{s}_j(v').
\end{equation*}
This concludes the induction. In the special case where $v=\varepsilon$ we obtain:
\begin{equation*}
f(w) = m_{\varepsilon} \lambda^*(C_1,w) \widehat{s}_i(\varepsilon)=m_{\varepsilon} \lambda^*(C_1,w) \Psi(C_i) = O_{\cal T}(w).
\end{equation*}
where $C_i=\delta^*(C_1,w)$. Therefore $O_{\cal T}=f$ as required. \qed
\end{proof}
In the sequel we prove Lemma~\ref{lemma:Nerode-Myhill}.
The proof below assumes that all $\alpha_i$ are $f$-essential. However, it can be easily amended to handle the general case.
We add a remark on this after the proof.

We start with the following simple lemma:
\begin{lemma}\label{finite_case}
For each $i\le n$, there is a partial function $s_i:\Sigma^*\rightarrow M$ and an element $m_i''\in M$
such that for every $\beta\in A_i$ there is a witness $\tuple{m_{\beta},m_i'',s_i}$ for $\beta\equiv_{R_f} \alpha_i$.
\end{lemma}
\begin{proof}
For each $\beta\in A_i$ there is a witness $\tuple{m'_{\beta},m_{\beta}'',s_{\beta}}$ for $\beta\equiv_{R_f} \alpha_i$. Let $\beta_i\in A_i$ be such that:
\begin{equation*}
\|m_{\beta_i}'' \| = \max_{\beta\in A_i} m_{\beta}''. 
\end{equation*}
Since $A_i$ is finite, $\beta_i$ exists. We set $s_i=s_{\beta_i}$ and $m_i''=m_{\beta_i}''$.
Now we have that for every $w$ such that $\alpha_i w\in \Dom(f)$:
\begin{equation*}
f(\alpha_i w) = m_{\beta}'' s_{\beta}''(w)= m_{\beta_i}'' s_{\beta_i}(w) =m_i'' s_i(w).
\end{equation*}
The last equality, along with $\|m_{\beta}''\| \le \|m_i''\|$ shows that $m''_{\beta}\le_M m_i''$. Hence,
$m_i''=m_{\beta} b_{\beta}$ for some $b_{\beta}\in M$. Therefore $s_{\beta}''(w)=b_{\beta} s_{i}(w)$ for all $\alpha_iw \in \Dom(f)$.
Consequently for each $\beta\in A_i$:
\begin{equation*}
f(\beta w)=m'_{\beta} s_{\beta}(w)=m'_{\beta}b_{\beta} s_i(w).
\end{equation*} 
This shows that $\tuple{m'_{\beta}b_{\beta},m''_i,s_i}$ is a witness for $\beta\equiv_{R_f} \alpha_i$ for all $\beta\in A_i$.
\end{proof}

Next, we proceed stepwise to define the required functions, $\widehat{s}_i$, for all the equivalence classes $C_i$. 
The idea is to start with big enough, but finite, subsets of each of the classes and greedily define a uniform witness-function 
$s_i$ for the selected finite part of $C_i$. Next, we shall use the right invariance of the equivalence relation, $R_f$, and consider
the functions $s_i$ in ensemble. Finally, using the properties of the norm and Lemma~\ref{finite_case}, we shall prove that
appropriate $\widehat{s}_i$ can be defined as $\widehat{s}_i(w)=s_j(z_i w)$.

We start by noting that since $\equiv_{R_f}$ is right-invariant and has index $n$, 
each equivalence class $C_i$ contains a representative of length less than $n$.
Indeed, consider the deterministic automaton with states $C_i$ and transitions
 $\delta(C_i,a)=C_j$ if and only if $C_i a\subseteq C_j$. Clearly, $\alpha\in C_i$
 if and only if:
 \begin{equation*}
 \delta^*([\varepsilon]_{\equiv_{R_f}},\alpha)=[\alpha]_{\equiv_{R_f}}.
 \end{equation*}
Since the sets $C_i$ are nonempty, any $C_i$ can be reached from $[\varepsilon]_{\equiv_{R_f}}$
via a simple path, in particular, it contains a word of length less than $n$.

Thus, we can and we do assume that $\alpha_i\in C_i$ are representatives of $C_i$ with $|\alpha_i|<n$ for each $i\le n$.
Let:
\begin{equation*}
A_i = C_i \cap \Sigma^{\le 2n}.
\end{equation*}
In particular, $\alpha_i \in A_i$. 
\begin{lemma}\label{lex_min}
Consider words $\alpha,\beta\in \Sigma^*$ s.t. $\alpha$ is $f$-essential and 
$\alpha\equiv_{f} \alpha\beta$. If $\tuple{m_1,m_2,s}$ witnesses for
$\alpha\equiv_{f} \alpha\beta$ then $m_1\le_M m_2$.
\end{lemma}
\begin{proof}
Since $\alpha$ is $f$-essential there is some $w$ s.t. $f(\alpha w)$ is defined. Since
$\equiv_{R_f}$ is right-invariant and $\alpha\equiv_{R_f} \alpha\beta$ we conclude that for every integer $n$
$\alpha\equiv_{R_f} \alpha\beta^n$. In particular $\alpha\beta^n w\in \Dom(f)$ for each $n$.

Since $\tuple{m_1,m_2,s}$ is a witness for $\alpha\equiv_{R_f} \alpha\beta$ we deduce that for each $n>0$:
\begin{equation*}
f(\alpha\beta^n w)=m_1 s(\beta^nw) \text{ and } f(\alpha\beta^{n}w)=m_2 s(\beta^{n-1}w).
\end{equation*}
Therefore, by the properties of the norm, we get that for each $n>0$:
\begin{equation*}
\| m_1 \| + \| s(\beta^nw)\|=\| m_2 \| + \| s(\beta^{n-1}w)\|.
\end{equation*}
Summing up this equalities for $n=1,2,\dots, N$ we obtain:
\begin{equation*}
N \|m_1\| + \|s(\beta^Nw)\| = N \|m_2\| + \|s(w)\|.
\end{equation*}
In particular,
\begin{equation*}
N (\|m_1\|-\|m_2\|) + \|s(\beta^Nw)\| = \|s(w)\|.
\end{equation*}
Since $\|s(\beta^Nw)\|\ge 0$ for each $N$ and $\|s(w)\|$ is constant we must have that
$\|m_1\|\le \|m_2\|$. Finally, since:
\begin{equation*}
m_1 s(\beta w)=m_2 s(w)
\end{equation*}
it must be that either $m_1\le_M m_2$, or $m_2<_M m_1$. However the latter is incompatible with 
$\|m_1\|\le  \|m_2\|$. Hence, $m_1\le_M m_2$.\qed
\end{proof}

For each $i\le n$ we pick a word $w_i\in \Dom(s_i)$. Since, we assume that all $\alpha_i$ are $f$-essential such
words exist. Next, we define a subsequential transducer over $\mathbb{R}$, 
$G=\tuple{\Sigma\times \mathbb{R},\{C_i\}_{i=1}^n,C_1,\emptyset,\delta_R,\lambda_R,\emptyset}$, where:
\begin{eqnarray*}
\delta_R(C_i,a) &=& C_j \iff \alpha_i a\equiv_{R_f} \alpha_j \\
\lambda_R(C_i,a) & = & \|s_i(aw_j)\| - \|s_j(w_j)\| \text{ for } \delta_R(C_i,a)=C_j\\
\end{eqnarray*}
The crucial property of $G$ is the following. 
\begin{claim}
Let $i,j\le n$, and $z,w\in \Sigma^*$ be such that $\delta_R^*(C_i,z)=C_j$ and $\alpha_j w\in \Dom(f)$.
Then:
\begin{equation*}
\lambda_R^*(C_i,z) = \|s_i(zw)\| - \|s_j(w)\|.
\end{equation*}
\end{claim}
\begin{proof}
The statement is obvious if $z=\varepsilon$. Let us consider the special case $z=a\in \Sigma$.
Then, $\alpha_i a\in C_j$ and since $|\alpha_i|<n$ we get that $|\alpha_i a|\le n\le 2n$. Hence
$\alpha_i a\in A_j$. Therefore there is a witness $\tuple{m,m''_j,s_j}$ for $\alpha_i a\equiv_{R_f} \alpha_j$.
Therefore for any $w$ s.t. $\alpha_j w\in \Dom(f)$:
\begin{equation*}
f(\alpha_i a w) = ms_j(w).
\end{equation*}
On the other hand we have that $f(\alpha_i a w)=m_i'' s_i(aw)$. This shows that
for every $\alpha_j w\in \Dom(f)$:
\begin{equation*}
 ms_j(w)= f(\alpha_i a w) = m_i'' s_i(aw).
\end{equation*}
Hence $\|m\| - \|m_j''\| = \|s_i(aw)\| -\|s_j(w)\|$. In the special case where $w=w_j$ we get:
\begin{equation*}
\lambda_R(C_i,a) = \|s_i(aw_j)\| - \|s_j(w_j)\| =\|m\| - \|m_j''\|.
\end{equation*}
Therefore $\lambda_R(C_i,a)=\|s_i(aw)\| -\|s_j(w)\|$ for all $w$ with $\alpha_j w\in \Dom(f)$.
Now the conclusion follows by straightforward induction on the length of $z$. For $z=\varepsilon$ there is nothing
to prove. Let $z=z' a$ and $\delta_R^*(C_i,z')=C_k$ and $\delta_R(C_k,a)=C_j$. If $\alpha_j w\in \Dom(f)$,
then $\alpha_k a w\in \Dom(f)$. Therefore by the induction hypothesis we have:
\begin{equation*}
\lambda_R^*(C_i,z') = \|s_i(z' a w)\| - \| s_k(aw)\|. 
\end{equation*}
Now by the special case we considered above, we get:
\begin{equation*}
\lambda_R(C_k,a) = \|s_k(a w)\| - \|s_j(w)\|.
\end{equation*} 
Summing up we obtain:
\begin{equation*}
\lambda_R^*(C_i,z) = \lambda_R^*(C_i,z')+\lambda_R(C_k,a)=\|s_i(z' a w)\|- \|s_j(w)\|=\|s_i(z w)\|- \|s_j(w)\|
\end{equation*}
as required. \qed
\end{proof}

\begin{corollary}
For each cycle $\delta_R^*(C_i,z)=C_i$ in $G$, $\lambda_R^*(C_i,z)\ge 0$.
\end{corollary}
\begin{proof}
Assume that the statement were not true. Then, there is a negative cycle in $G$.
Therefore there is also a simple negative cycle in $G$. Since $G$ has $n$ states, there is some $z\in \Sigma^{\le n}$
and $C_i$ such that $\lambda_R^*(C_i,z)<0$. Now we have that $\alpha_i\in C_i$ and $|\alpha_i|<n$.
Hence $\alpha_i z\in C_i$ and $|\alpha_i z|<2n$. Therefore $\alpha_i z\in A_i$. By the construction of $s_i$,
there is a witness $\tuple{m,m_i'',s_i}$ for $\alpha_i z\equiv_{R_f} \alpha$. By Lemma~\ref{lex_min} we have that
$\|m\|\ge \|m_i''\|$. On the other hand for each $w$ with $\alpha_i w\in \Dom(f)$ we have also $\alpha_i z w\in \Dom(f)$ and thus:
\begin{equation*}
m_i'' s_i(zw) = f(\alpha_i z w) = m s_i(w).
\end{equation*}
Hence $\|s_i(zw)\|\ge \|s_i(w)\|$. But by the claim above we have that:
\begin{equation*}
\lambda^*(C_i,z) = \|s_i(z w)\|- \|s_i(w)\|\ge 0.
\end{equation*}
Thus, $\lambda^*(C_i,z)$ is not negative contradicting the assumption for a negative cycle in $G$. \qed
\end{proof}

Now we are ready to prove Lemma~\ref{lemma:Nerode-Myhill}
\begin{proof}[Lemma~\ref{lemma:Nerode-Myhill}]
Let $\alpha_i$, $A_i$, $s_i$, and $G$ be as above. For each $i$ we define:
\begin{equation*}
\tuple{j_i,z_i} =\arg\min_{\tuple{j,z}} \{\lambda_R^*(C_j,z)\,|\, \delta_R^*(C_j,z)=C_i\}.
\end{equation*}
Since $G$ contains no negative cycles all the pairs $\tuple{j_i,z_i}$ are well defined. By the
same argument all the minima are attained also for words $z'_i\in \Sigma^{<n}$. Thus, we can and
we do assume that $|z_i|<n$ for all $i$. 
Note that by the Claim for $G$ we have that if $\delta^*(C_j,z)=C_i$, then:
\begin{equation*}
\lambda_R^*(C_j,z) = \|s_j(zw)\| - \|s_i(w)\|.
\end{equation*}
for any $\alpha_i w\in \Dom(f)$. 
Thus if $\alpha_i w\in \Dom(f)$, minimising $\lambda_R^*(C_j,z)$ subject to $\delta_R^*(C_j,z)=C_i$ is the same as minimising $\|s_j(zw)\|$.
We define $\widehat{s}_i$ as:
\begin{equation*}
\widehat{s}_i(w) = s_{j_i}(z_i w) \text{ for all }i\le n.
\end{equation*} 
By the remark above we have that $\|\widehat{s}_i(w)\| \le \|s_i(w)\|$.

We prove that $\widehat{s}_i$ satisfy the required properties. First we show that for any $i,j\le n$ and character $a\in \Sigma$ with $\alpha_i a \equiv_{R_f} \alpha_j$ there is a monoid element
$m_{i,a,j}\in M$ with:
\begin{equation*}
\widehat{s}_i(aw) =m_{i,a,j}\widehat{s}_j (w) \text{ whenever } \alpha_j w \in \Dom(f).
\end{equation*}
Indeed, since$\alpha_{j_i} z_i \in C_i$ and $|\alpha_{j_i} z_i|< n + n=2n$, we have that $\alpha_{j_i} z_i \in A_i$.
Therefore, by Lemma~\ref{finite_case} there is some element $m_i$ such that
$\tuple{m_i,m_i'',s_i}$ is a witness for $\alpha_{j_i}z_i \equiv_{R_f} \alpha_i$. Hence for all $w$ such that
$\alpha_{j_i} z_i w\in  \Dom(f)$ we have:
\begin{equation*}
f(\alpha_{j_i} z_i w)=m_i s_i(w). 
\end{equation*}
On the other hand by the definition of $m_{j_i}''$ and $s_{j_i}$ we have that:
\begin{equation*}
f(\alpha_{j_i} z_i w_i)=m_{i_j}'' s_{i_j}''(z_iw)=m_{i_j}'' \widehat{s}_i(w). 
\end{equation*}
By these two equalities, and taking into account that $\|\widehat{s}_i(w)\|\le \|s_i(w)\|$
for all $\alpha_iw\in \Dom(f)$, by Lemma~\ref{RightCancellation} we conclude that there is some $b_i$ such that:
\begin{equation*}
s_i(w) = b_i \widehat{s}_i(w) \text{ for } \alpha_i w\in \Dom(f).
\end{equation*}
Consequently for each $i\le n$ and each $\beta\in A_i$ there is a witness $\tuple{m_{\beta}b_i,m_i'' b_i,\widehat{s}_i}$

Now, since $\alpha_i \in C_i$ and $\alpha_i a\in C_j$ is of length $|\alpha_i a| \le n< 2n$, we get that
$\alpha_i a\in A_j$. Therefore there is a witness $\tuple{m,m''_jb_j,\widehat{s}_j}$ for $\alpha_i a\equiv_{R_f} \alpha_j$.
Therefore for all $\alpha_i aw \in \Dom(f)$ we have:
\begin{equation*}
f(\alpha_i a w) = m \widehat{s}_j(w) \text{ and } f(\alpha_i a w)=m''_i b_i \widehat{s}_i(aw).
\end{equation*}
Since $\|\widehat{s}_j(w)\| \le \|\widehat{s}_i(aw)\|$, by Lemma~\ref{RightCancellation}, we obtain that $m_i''\le_M m$. Let $m_{i,a,j}$ be such that:
\begin{equation*}
m_i'' m_{i,a,j}=m.
\end{equation*}
Thus, for all $w$ such that $\alpha_j w\in \Dom(f)$ we have:
\begin{equation*}
\widehat{s}_i(aw)= m_{i,a,j} \widehat{s}_j(w).
\end{equation*} 

Finally, we prove that for each $i$, for each $\beta\in C_j$ there is a witness $\tuple{m_{\beta},m''_j b_j, \widehat{s}_j}$ for
$\beta\equiv_{R_f} \alpha_j$. We proceed by induction. This is obvious in the case where $|\beta|\le 2n$, since in this case $\beta\in A_j$.
If $\beta=\beta' a$, then there is some $i$ with $\beta' \equiv_{R_f} \alpha_i$. By the induction hypothesis there is a witness $
\tuple{m_{\beta'},m_i'' b_i,\widehat{s}_i}$ for $\beta'\equiv_{R_f} \alpha_i$. Thus for each $w$ with $\alpha_j w\in \Dom(f)$ we have:
\begin{equation*}
f(\beta' a w) = m_{\beta'}\widehat{s}_i(aw)= m_{\beta'}m_{i,a,j}\widehat{s}_j(w).
\end{equation*}
This shows that $\tuple{m_{\beta'}m_{i,a,j},m_j''b_j,\widehat{s}_j}$ is a witness for $\beta \equiv_{f} \alpha_j$. \qed
\end{proof}

\begin{remark}
The construction of the subsequential transducer $G$ assumes that there are witnesses $w_i$ for each class $C_I$,
i.e. that every word is $f$-essential. However, we can easily amend this proof to the general case by noting that there
is at most one class $C_{err}$ that contains exactly those words that are not $f$-essential. Furthermore $C_{err}$ will
have arcs leading only to $C_{err}$. Thus, we can safely remove the class $C_{err}$ from the construction of $G$
and endow the remaining classes $C_i$ with the appropriate $\widehat{s}_i$. Adding $\widehat{s}_{err}=\emptyset$ does
not spoil the conclusion of Lemma~\ref{lemma:Nerode-Myhill} since there are no witnesses for the words in $C_{err}$
and no arcs leave this class.
\end{remark}


\section{Discussion and Conclusion}
In the paper we have considered the Nerode-Myhill relation for functions over sequentiable structures
and proved that any function of finite index is a subsequential rational function. This generalises the
results of Mohri~\cite{Mohri00} for the case of free monoids and the monoid of nonnegative real numbers with addition.

\bibliographystyle{splncs03}
\bibliography{lata2107}
\end{document}